\documentclass[journal]{IEEEtran}
\usepackage[utf8]{inputenc}
\usepackage{amsmath,amsfonts,amsthm}
\usepackage{algorithm,algorithmic}
\usepackage{booktabs}
\usepackage{bm}
\usepackage{comment}
\usepackage{chngpage}
\usepackage{enumerate}
\usepackage{paralist}
\usepackage{extramarks}
\usepackage{fancyhdr}
\usepackage{graphicx,float,wrapfig}
\usepackage{lastpage}
\usepackage{setspace}
\usepackage{soul,color}
\usepackage{xspace}
\usepackage{indentfirst}

\usepackage{nohyperref}

\makeatletter
\DeclareRobustCommand\onedot{\futurelet\@let@token\@onedot}
\def\@onedot{\ifx\@let@token.\else.\null\fi\xspace}

\def\eg{\emph{e.g}\onedot} 
\def\ie{\emph{i.e}\onedot} \def\Ie{\emph{I.e}\onedot}

\def\st{$st~$}
\makeatother

\theoremstyle{definition}
\newtheorem{defn}{Definition}

\theoremstyle{plain}
\newtheorem{thm}{Theorem}
\newtheorem{lem}{Lemma}

\newtheorem{crl}{Corollary}

\theoremstyle{remark}
\newtheorem*{rmk}{Remark}

\newcommand{\CC}{\mathcal{C}}

\newcommand*{\rom}[1]{\expandafter\@slowromancap\romannumeral #1@}

\newif\iftechrep \techrepfalse

\catcode`T=12 \catcode`R=12	
\def\checkiftechreport#1{
\expandafter\iistechreport#1TR. \techreptrue\fi}
\def\iistechreport#1TR#2.{\def\tmp{#2}\ifx\tmp\empty\else}
\catcode`T=11 \catcode`R=11
\def\checkTR{\checkiftechreport{\jobname}}
\checkTR

\title{Connectivity in Interdependent Networks}
\author{Jianan Zhang, and Eytan Modiano \thanks{The authors are with the Laboratory for Information and Decision Systems, Massachusetts Institute of Technology. This work was supported in part by DTRA grants HDTRA1-13-1-0021 and HDTRA1-14-1-0058.}}
\begin{document}
\maketitle
\begin{abstract}
We propose and analyze a graph model to study the connectivity of interdependent networks. Two interdependent networks of arbitrary topologies are modeled as two graphs, where every node in one graph is supported by supply nodes in the other graph, and a node fails if all of its supply nodes fail. Such interdependence arises in cyber-physical systems and layered network architectures. 

We study the \emph{supply node connectivity} of a network: namely, the minimum number of supply node removals that would disconnect the network. We develop algorithms to evaluate the supply node connectivity given arbitrary network topologies and interdependence between two networks. Moreover, we develop interdependence assignment algorithms that maximize the supply node connectivity. We prove that a random assignment algorithm yields a supply node connectivity within a constant factor from the optimal for most networks. 
\end{abstract}

\section{Introduction}
The development of smart cities and cyber-physical systems has brought interdependence between once isolated networks and systems. In interdependent networks, one network depends on another to achieve its full functionality. 
Examples include smart power grids \cite{rosato2008modelling, parandehgheibi2013robustness}, transportation networks \cite{gu2011onset, yagan2012optimal}, and layered communication networks \cite{ghani2000ip, lee2011cross}. Failures in one network not only affect the network itself, but also may cascade to another network that depends on it. For example, in the Italy blackout in 2003, an initial failure in the power grid led to reduced functionality of the communication network, which led to further failures in the power grid due to loss of communication and control \cite{rosato2008modelling, buldyrev2010catastrophic}. Thus, the robustness of a network relies on both its own topology and the interdependence between different networks.

Interdependent networks have been extensively studied in the statistical physics literature based on random graph models since the seminal work of \cite{buldyrev2010catastrophic}. Nodes in two random graphs are interdependent, and a node is functional if both itself and its interdependent node are in the largest component of their respective graphs. If a positive fraction of nodes are functional as the total number of nodes approaches infinity, the interdependent random graphs percolate. The condition for percolation measures the robustness of the interdependent networks. While these models are analytically tractable, percolation may not be a key indicator for the functionality of infrastructure networks. For example, a network would lose most of its functionality when a large fraction of nodes are removed, while the graph still percolates.

A few models have been proposed for specific applications to capture the dependence between networks, such as interdependent power grids and communication networks \cite{parandehgheibi2013robustness, parandehgheibi2015modeling}, and IP-over-WDM networks \cite{lee2011cross, lee2014maximizing}. These models consider finite size, arbitrary network topology, and incorporate dynamics in real-world networks. Instead of percolation, more realistic metrics are used to capture the robustness of interdependent networks, such as the amount of satisfied power demand, or traffic demand. These models are able to capture important performance metrics in real-world networks, at the cost of more complicated modeling and analysis.

We develop an analytically tractable model for interdependent networks which aims to capture key robustness metrics for infrastructure networks. In contrast to the random graph models where some assumptions are difficult to justify in infrastructure networks (\eg, very large network size and randomly placed links), we use a deterministic graph model to represent each network, where nodes and edges are specified by the topology of an infrastructure network. 
We develop metrics that measure the robustness of interdependent networks, by generalizing canonical metrics for the robustness of a single network. Moreover, our model is simple enough to allow for the evaluation of the robustness of interdependent networks, and allows us to obtain insights and principles for designing robust interdependent networks. 

\subsection{Related work}
A closely related model is the shared risk group model \cite{srg2, hu2003diverse,coudert2007shared,lee2011cross}, where a set of edges or nodes share the same risk and can be removed by a single failure event. The model is used to study the robustness in layered communication networks such as IP-over-WDM networks. In interdependent networks, multiple demand nodes in one network may depend on the same node in another network, and they share the same risk (of the supply node's failure). Suppose that a demand node has multiple supply nodes, and is content to have at least one supply node. The interdependent networks can be viewed as a generalized shared risk group model, given that the occurrences of multiple risks, instead of one single risk, are required to remove a node in the interdependent networks.

The shared risk group model can be represented by a colored graph (or labeled graph), in which edges or nodes that share the same risk have the same color (or label) \cite{coudert2007shared, yuan2005minimum, klein2016colored}. Complexity results and approximation algorithms have been developed to compute the minimum number of colors that appear in an \emph{edge cut} that disconnects a colored graph \cite{coudert2007shared, zhang2011approximation}. In interdependent networks, we study node failures due to the removals of their supply nodes. Thus, our focus is on the \emph{node cut} in a colored graph with colored nodes and regular edges. While most results for edge cuts that separate a pre-specified source-destination pair (\ie, \st edge cuts) can be naturally extended to \st node cuts, the extension is not obvious when the global edge or node cuts of a graph are considered. Although it is possible to transform a node cut problem in an undirected graph into an edge cut problem in a directed graph, the nature and analysis of the problem in a directed graph are different from the problem in an undirected graph, when global cuts are considered \cite{karger2001randomized, censor2014distributed}. Thus, new techniques need to be developed in this paper to study the global node cuts in a colored graph.

While most studies on the shared risk group model have focused on the evaluation of robustness metrics of a given network, there have also been previous works that take a \emph{network design} approach to optimize the metrics. For example, in optical networks, where two logical links share the same risk if they are supported by the same physical link, \cite{lee2011cross, lee2014maximizing} developed lightpath routing algorithms that maximize the number of physical link failures that a given logical topology can tolerate. In this paper, we study the interdependence assignment that maximizes the number of supply node failures that a network can tolerate (to stay connected). Instead of solving difficult integer programs as in most network design literature, we apply graph algorithms, \eg, the vertex sampling and graph partitioning techniques \cite{censor2014new, Censor-Hillel2015}, to develop polynomial time algorithms that have provable performance guarantees. The vertex sampling techniques provide bounds on the probability that the graph is connected after random node removals. We build connections between the node removals in a single graph and the node failures in interdependent networks, and study the connectivity of interdependent networks.

\subsection{Our contributions}
We propose an analytically tractable model for two interdependent networks, and study the impacts of node failures in one network on the other network. We add a minimal ingredient to the classical graph model to capture interdependence, and define supply node connectivity as a robustness metric for our model, analogous to the widely accepted cut metric (node connectivity) for the classical graph model.
We prove the complexity, and develop integer programs to evaluate the supply node connectivity, both for a given pair of nodes and for the entire network. Moreover, we propose a polynomial time algorithm that computes the supply node connectivity for a special class of problems, based on which we develop an approximation algorithm for the general problem.

In addition, we study the network design problem of improving the robustness of interdependent networks by assigning interdependence between two networks. We propose a simple assignment algorithm that maximizes the supply node connectivity of an \st pair, by assigning node-disjoint paths with different supply nodes while allowing nodes in the same path to have the same supply node. Based on a similar idea and considering disjoint connected dominating sets, we develop an assignment algorithm that approximates the optimal global supply node connectivity to within a polylogarithmic factor. Finally, we propose a random assignment algorithm under which, with high probability, the global supply node connectivity is within a constant factor from the optimal in most cases, and at worst is within a logarithmic factor from the optimal.

The rest of the paper is organized as follows. In Section \ref{sc:model}, we develop a one-way dependence model, where a demand network depends on a supply network. This allows us to deliver key results and intuitions for studying the impacts of node failures in one network on its interdependent network, using simplified notations and presentations. We study this one-way dependence model in Sections \ref{sc:evaluate} and \ref{sc:assign}. In Section \ref{sc:evaluate}, we evaluate the supply node connectivity of the demand network. In Section \ref{sc:assign}, we develop algorithms, which assign supply nodes to demand nodes, to maximize the supply node connectivity. In Section \ref{sc:twoway}, we focus on the bidirectional interdependence model and generalize the above results. Section \ref{sc:simu} provides simulation results. Finally, Section \ref{sc:conclude} concludes the paper.

\section{One-way dependent network model and colored graph representation}
\label{sc:model}
\subsection{One-way dependence model}
We start by considering a one-way dependence model, where nodes in a \emph{demand network} depend on nodes in a \emph{supply network}. This simplified model allows us to focus on the impacts of node failures in one network on the other network. Let two undirected graphs $G_1(V_1, E_1)$ and $G_2(V_2, E_2)$ represent the topologies of the demand and supply networks, respectively. Each node in the demand network depends on one or more nodes in the supply network. The dependence is represented by the directed edges in Fig. \ref{fig:oneway}. Every supply node provides substitutional supply to the demand nodes. A demand node is functional if it is adjacent to at least one supply node. Figure \ref{fig:oneway} illustrates the failure of a demand node due to the removals of its supply nodes.

As a more concrete example, we use $G_1$ to represent a communication network and $G_2$ to represent a power grid. Each node in $G_1$ represents a router, and each node in $G_2$ represents a power station. A router receives power from one or more power stations, and fails if all of the supporting power stations fail.

\begin{figure}[h]
\begin{centering}
\leavevmode\includegraphics[width=0.75\linewidth]{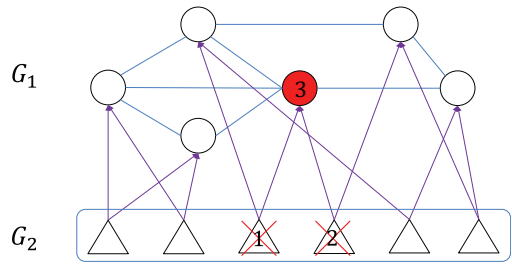}
\caption{Demand node 3 fails if both supply nodes 1 and 2 fail.}
\label{fig:oneway}
\end{centering}
\end{figure}

We aim to characterize the impacts of node removals in the supply network on the connectivity of the demand network. Recall that (see, \eg \cite{even1975network}), in a single graph, a \emph{node cut} (\ie, vertex cut) is a set of nodes whose removals either disconnect the graph into more than one connected component, or make the remaining graph trivial (where a single node remains). The \emph{node connectivity} of a graph is the number of nodes in the smallest node cut. In the one-way dependence model, the connectivity of the demand network depends not only on its topology $G_1(V_1, E_1)$, but also on the supply-demand relationship. We define the \emph{supply node cut} and \emph{supply node connectivity} of the demand network as follows.
\begin{defn}
A \emph{supply node cut} of the demand graph is a set of supply nodes whose removals induce a node cut in the demand graph. (Mathematically, a supply node cut of $G_1$ is a set of nodes $V_s \subseteq G_2$, such that nodes $V_d \subseteq G_1$ do not have any supply nodes other than $V_s$, and that $V_d$ contain a node cut of $G_1$.)

The \emph{supply node connectivity} is the number of nodes in the smallest supply node cut.
\end{defn}

The above definition is a generalization of the traditional node cut to include a superset of a cut. This is necessary because the removals of supply nodes may not correspond to proper cuts of the demand graph (see Fig. \ref{NodeCut}). Under this definition, graphs with larger supply node connectivity are more robust under supply node failures.

\begin{rmk}
In Fig. \ref{NodeCut}, suppose that every node has a single supply node, and that the red nodes share the same supply node $u \in G_2$. By removing $u$, the left graph stays connected after removing all the three red nodes, while the right graph is disconnected. However, the left graph is less robust under the removal of supply node $u$, because the failed nodes in the left graph include all the failed nodes in the right graph. Thus, ``graph connectivity after supply node removals'' does not serve as a good measure for the robustness of the demand graph when supply nodes fail. This motivates our definition of supply node cut and supply node connectivity. According to our definition, the supply node connectivity of the left graph is one.
\end{rmk}

\begin{figure}[h]
\begin{centering}
\leavevmode\includegraphics[width=.9\linewidth]{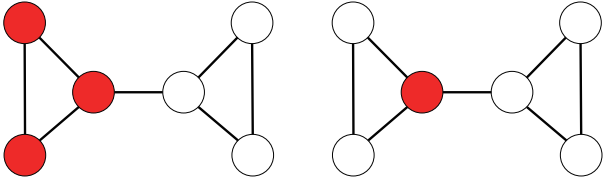}
\caption{Let the three red nodes in the left figure be supported by the same supply node. Removing the supply node leads to the failure of the three red nodes, which do not form a proper cut but form a superset of a proper cut (\ie, the red node in the right figure). The supply node is viewed as a supply node cut.}
\label{NodeCut}
\end{centering}
\end{figure}

We study the connectivity of a source-destination pair $(s,t) \in G_1$ as a starting point, which provides insights towards the graph connectivity with simpler analysis. In a graph, an \emph{$st$ node cut} is a set of nodes, excluding $s$ and $t$, whose removals disconnect $s$ from $t$. The number of nodes in the smallest \st node cut is the \emph{\st node connectivity}. Analogously, we define \emph{\st supply node cut} and \emph{\st supply node connectivity} as follows.
\begin{defn}
  An \emph{\st supply node cut} is a set of supply nodes whose removals induce an \st node cut. (Mathematically, an \st supply node cut is a set of nodes $V_s^{st} \subseteq G_2$, such that nodes $V_d^{st} \subseteq G_1$ do not have any supply nodes other than $V_s^{st}$, and that $V_d^{st}$ contain an \st node cut.)

  The \emph{\st supply node connectivity} is the number of nodes in the smallest \st supply node cut.
\end{defn}

An \st supply node cut may induce demand node failures $V_d^{st}$ including $s$ and/or $t$, since $s$, $t$ may share the same supply nodes with nodes in the \st node cut. However, removing $V_d^{st} \setminus \{s,t\}$ must disconnect $s$ from $t$.

We consider non-adjacent $s$ and $t$ throughout the paper. Otherwise, if $s$ and $t$ are adjacent, they are always connected when other nodes are removed, and there is no node cut that disconnects them.

\subsection{Transformation to a colored graph}
\label{sc:tranColor}
Our model is closely related to the shared risk node group (SRNG) model \cite{coudert2007shared,datta2004diverse}. In the SRNG model, several nodes share the same risk, and can be removed by a single failure event. In interdependent networks, if every node has \emph{one} supply node, then the demand graph becomes exactly the same as the SRNG model, where the demand nodes that have the same supply node share the same risk.

The SRNG model can be represented by a \emph{colored graph}, where the nodes that have the same color share a common risk. We define\footnote{Previous studies on colored graphs focused on color edge cuts in colored graphs with colored edges and regular nodes. Much less is known about the color node cut, a counterpart of color edge cut, in colored graphs with colored nodes and regular edges. In fact, to the best of our knowledge, there is no formal definition for color node cut.} \emph{color node cut} and \emph{\st color node cut} as follows.
\begin{defn} \label{def:colorglobal}
  Given a colored graph $G(V,E,\CC)$ with colored nodes $V$, regular edges $E$, and node-color pairs $\CC$ that represent the color for each node, a \emph{color node cut} is a set of colors $C_c$ such that the nodes covered by colors $C_c$ contain a node cut of $G$.

  A \emph{minimum color node cut} of $G$ is a color node cut $C_{c \min}$ that has the minimum number of colors. The number of colors in $C_{c \min}$ is the \emph{value} of the minimum color node cut.
\end{defn}

\begin{defn}
\label{def:colorst}
  Given a colored graph $G(V,E,\CC)$ with colored nodes $V$, regular edges $E$, node-color pairs $\CC$ that represent the color for each node, and a pair of nodes $(s,t) \in V$, a \emph{color \st node cut} is a set of colors $C_c^{st}$ such that the nodes covered by colors $C_c^{st}$ contain an \st node cut. 

  A \emph{minimum color \st node cut} is a color \st node cut $C_{c \min}^{st}$ that has the minimum number of colors. The number of colors in $C_{c \min}^{st}$ is the \emph{value} of the minimum color \st node cut.
\end{defn}


Colored graph provides an intuitive representation of the correlated node failures using color. If every demand node has a single supply node, then every demand node has a color that corresponds to its supply node. After the failure of a supply node, a demand node fails if it has the color that corresponds to the supply node.

In general, a demand node can have multiple supply nodes, and thus the mapping to a colored graph is not straightforward. We propose Algorithm \ref{al:tran} that transforms the demand network to a colored graph where every node has a single color, and use Fig. \ref{fig:transformation} to illustrate the algorithm. 

\begin{algorithm}[h]
\caption{Transformation from the demand graph $G_1$ to a colored graph $\tilde G_1$.}
\label{al:tran}
\begin{enumerate}
\item If a node $v_i \in G_1$ has $n_s(v_i)$ supply nodes, $n_s(v_i)$ copies of $v_i$ exist in $\tilde G_1$. Each copy has a color which identifies a supply node. No edge exists between the copies of $v_i$.
\item If $v_i$ and $v_j$ are connected by an edge in $G_1$, then all the copies of $v_i$ are connected to all the copies of $v_j$ in $\tilde G_1$.
\end{enumerate}
\end{algorithm}

\begin{figure}[h]
\begin{centering}
\leavevmode\includegraphics[width=\linewidth]{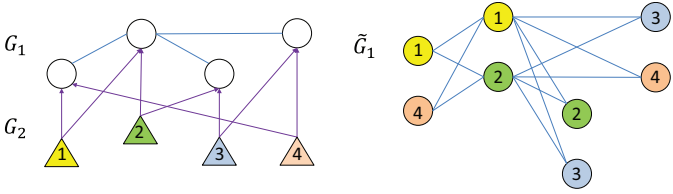}
\caption{Illustration of the transformation algorithm. }
\label{fig:transformation}
\end{centering}
\end{figure}

We study the connectivity of the demand graph based on the colored graph, due to the following theorem.
\begin{thm} \label{th:mapping}
There is a one-to-one mapping between a supply node cut in the demand network and a color node cut in the transformed graph of the demand network.
\end{thm}

\begin{proof}

Let $G_1$ be the demand graph and $G_2$ be the supply graph.
Let $\tilde G_1$ be the transformed graph of $G_1$ by Algorithm \ref{al:tran}. The result trivially holds if every demand node has a single supply node. Next we focus on the case where a demand node has more than one supply node.

We first prove that given any supply node cut $V_s$ of $G_1$, there exists a color node cut $C_c$ of $\tilde G_1$ where colors $C_c$ correspond to supply nodes $V_s$. According to the definition of a supply node cut, the demand nodes in $G_1$ that have no supply nodes other than $V_s$ contain a node cut $V^*_d$ of $G_1$. By removing $V^*_d$ from $G_1$, either $G_1$ is separated into at least two components, or a single node $v$ in $G_1$ remains (by the definition of a node cut for a graph). In the first case, nodes in $\tilde G_1$ that correspond to $V^*_d \subseteq G_1$ have colors in $C_c$ and they are removed. Among the remaining nodes, if no edge exists between two nodes in $G_1$, then there is no edge between their corresponding nodes in $\tilde G_1$. Therefore, the remaining nodes in $\tilde G_1$ are disconnected after removing the nodes that correspond to $V_d^*$ and have colors $C_c$.
In the second case, copies of $v$ are the only remaining nodes in $\tilde G_1$ and they are disconnected. Thus, $C_c$ is a color node cut in $\tilde G_1$ in both cases.

We then prove that given any color node cut $C_c$ of $\tilde G_1$, there exists a supply node cut $V_s$ of $G_1$ where $V_s$ corresponds to colors $C_c$. After removing all (or a subset) of nodes in $\tilde G_1$ that have colors $C_c$, either a single node remains in $\tilde G_1$, or $\tilde G_1$ is separated into multiple connected components. In the first case, at most a single node remains in $G_1$ after removing $V_s$, and thus $V_s$ is a supply node cut. In the second case, if every component contains a single node, and the node corresponds to the same node in $G_1$, then at most one node survives in $G_1$ by removing supply nodes $V_s$. On the other hand, if these components correspond to different nodes in $G_1$, there must exist two disconnected nodes $v_1, v_2 \in G_1$, whose copies are in different components in $\tilde G_1$. (Recall that, if two nodes are connected in $G_1$, then their copies are connected in $\tilde G_1$. If all the remaining nodes in $G_1$ form a connected component, then their corresponding copies in $\tilde G_1$ also form a connected component.) In both cases, $V_s$ is a supply node cut of $G_1$.
\end{proof}



Moreover, an \st supply node cut can be represented by a color $\tilde{s} \tilde{t}$ node cut in the colored graph, where $\tilde{s}$ is any copy of $s$ and $\tilde{t}$ is any copy of $t$.
By considering cuts that separate $(s,t)$ in $G_1$ and cuts that separate $(\tilde s,\tilde t)$ in $\tilde G_1$, we obtain the following result by a similar proof to that of Theorem \ref{th:mapping}.
\begin{crl}
\label{th:mappingst}
There is a one-to-one mapping between a supply node \st cut in the demand network and a color $\tilde{s} \tilde{t}$ node cut in the transformed graph of the demand network, where $\tilde{s}$ is any copy of $s$ and $\tilde{t}$ is any copy of $t$.
\end{crl}

Another corollary is a property of the transformed graph $\tilde G_1$ when every demand node in $G_1$ has a fixed number $n_s$ of supply nodes. If $G_1$ has $n_1$ nodes and $m_1$ edges, the transformed graph $\tilde G_1$ has $n_1 n_s$ nodes and $m_1 n_s^2$ edges. Moreover,
\begin{crl}
\label{th:cutvalue}
If every demand node has a fixed number $n_s$ of supply nodes, the following results hold.

If the node connectivity of $G_1$ is $k_1$, then the node connectivity of $\tilde G_1$ is $k_1 n_s$.

If the \st node connectivity is $k^{st}_1$ ($s,t \in G_1$), then the $\tilde{s} \tilde{t}$ node connectivity is $k^{st}_1 n_s$, where $\tilde{s} \in \tilde G_1$ is any copy of $s$ and $\tilde{t} \in \tilde G_1$ is any copy of $t$.
\end{crl}
\begin{proof}
  By assigning $n_s$ distinct supply nodes to each node in $G_1$, using a total of $n_1 n_s$ supply nodes, to remove a node in $G_1$, a distinct set of $n_s$ supply nodes must be removed. Thus, the supply node connectivity of $G_1$ equals the node connectivity of $G_1$ times $n_s$. Moreover, in $\tilde G_1$, every node has a distinct color, and the number of colors in a color node cut equals the number of nodes in the same node cut. Thus, the node connectivity of $\tilde G_1$, without considering colors, equals the supply node connectivity of $G_1$, because of the one-to-one mapping proved in Theorem \ref{th:mapping}. We have therefore proved that the node connectivity of $\tilde G_1$ is the node connectivity of $G_1$ times $n_s$. The same relationship holds for \st node connectivity in $G_1$ and $\tilde{s} \tilde{t}$ node connectivity in $\tilde G_1$.
\end{proof}

\subsection{Notations}
We define notations to be used throughout the rest of the paper. For a finite set $X$, the cardinality of $X$ is denoted by $|X|$. For a colored graph $G(V,E,\CC)$, the number of nodes, edges, and colors are denoted by $n,m,n_c$, respectively. The graph connectivity is denoted by $k$, and the \st connectivity is denoted by $k^{st}$. The subscript $i \in \{1,2\}$ denotes the identity of a graph. For example, $n_1$ denotes the number of nodes in $G_1$. The subscript $s$ denotes \emph{supply}. For example, $n_{s1}$ denotes the number of supply nodes for a node in $G_1$.

We use asymptotic notations in this paper. Let $f(x) > 0$ and $g(x) > 0$ be two functions. If there exists a constant $M$ and a positive number $x_0$, such that $f(x) \leq M g(x)$ for all $x \geq x_0$, then $f(x) = O(g(x))$. Moreover, $f(x) = \Omega(g(x))$ if $g(x) = O(f(x))$; $f(x) = \Theta(g(x))$ if both $f(x) = O(g(x))$ and $f(x) = \Omega(g(x))$; $f(x) = o(g(x))$ if $\lim_{x\rightarrow \infty} f(x) / g(x) = 0$; $f(x) = \omega(g(x))$ if $g(x) = o(f(x))$.

\section{Evaluation of the supply node connectivity}
\label{sc:evaluate}
In this section, we study the supply node connectivity of the demand network. As discussed in the previous section, supply node cuts in the demand network are equivalent to color node cuts in a colored graph. To simplify the presentation, we consider a colored graph $G(V,E,\CC)$ throughout this section.
\subsection{Complexity}
We prove that computing both the global minimum color node cut of a graph and the minimum color \st node cut are NP-hard.
The proof for the complexity of the minimum color \st node cut follows a similar approach to that of the minimum color \st edge cut in \cite{coudert2007shared}. In contrast, the complexity of the global minimum color edge cut is unknown. The detailed proofs of Theorems \ref{th:np} and \ref{th:npst} can be found in the appendix.

\begin{thm}\label{th:np}
  Given a colored graph, computing the value of the global minimum color node cut is NP-hard.
\end{thm}

\begin{thm} \label{th:npst}
  Given a colored graph and a pair of nodes $(s,t)$, computing the value of the minimum color $st$ node cut is NP-hard.
\end{thm}

Given the computational complexity, in the remainder of this section, we first develop integer programs to compute the exact values of the minimum color cuts, and then develop polynomial time approximation algorithms.
\subsection{Exact computation for arbitrary colored graphs}
We compute the minimum color \st node cut using a mixed integer linear program (MILP). In this formulation, each node has a \emph{potential}. Connected nodes have the same potential. The source and the destination are disconnected if they have different potentials. We note that the classical MILP formulation for computing the minimum \emph{edge cut} also uses node potentials to indicate disconnected components after removing edges \cite{bertsimas1997}.

In the MILP formulation, indicator variable $c_r$ denotes whether color $r \in C$ is in the minimum color cut, where $C$ is the set of colors in the colored graph. Indicator variable $y_v$ denotes whether node $v \in V$ is a cut node that separates the \st pair, and may take value 1 only if the color of $v$ is in the color cut. Note that $y_v$ may take value 0 even if the color of $v$ is in the color cut (constraint (\ref{milp:color})). This allows the cut nodes to be a \emph{subset} of nodes with colors $\{r | c_r = 1\}$ (recall Definition \ref{def:colorst}).

The potential of a node $v$ is denoted by $p_v$.
After removing all the cut nodes, the potentials of nodes in a connected component are the same, guaranteed by constraints (\ref{milp:cut1}) under the condition $y_i = y_j = 0$. The same constraints guarantee that nodes adjacent to the cut nodes may have different potentials from the cut nodes, if $y_i = 1$ or $y_j = 1$. The potential of the source is 0, and the potential of the destination is 1, guaranteed by constraint (\ref{milp:stpotential}). Moreover, constraint (\ref{milp:cutst}) guarantees that neither $s$ nor $t$ is a cut node. Thus, the component that contains $s$ and the component that contains $t$ are separated by an \st node cut. The objective is to minimize the number of colors of the cut nodes.
\begin{eqnarray}
\text{  min }   &&\sum_{r \in C} c_r ~~~~~~~~~~~~~~~~~~~~~~~~~~~~~~~~~~~~~~~~~~~~\text{(MILP)} \nonumber \\
\text{s.t.}
&& - y_i - y_j \leq p_i - p_{j} \leq y_i + y_j, ~~ \forall (i,j) \in E, \label{milp:cut1} \\
&& p_s = 0, p_t = 1, \label{milp:stpotential}\\
&& y_s = y_t = 0, \label{milp:cutst} \\
&& y_v \leq c_r, \forall r \in C, v \in \{v | r \text{ is the color of } v\}, \label{milp:color} \\
&& p_v, y_v \geq 0, ~~\forall v \in V, \nonumber \\
&& c_r \in \{0,1\}, ~~\forall r \in C  \nonumber.
\end{eqnarray}

Next we compute the global minimum color node cut of a colored graph using an integer program (IP). The variables $c,y,p$ have the same representations as those in the above MILP. 
Recall that a global node cut of a graph \emph{either} separates the remaining nodes into disconnected components, \emph{or} makes the remaining graph trivial. In the first case, $z = 0$, and constraint (\ref{milp:one1}) guarantees that there is at least one node with potential 1, in addition to all the cut nodes. Constraints (\ref{milp:cutpotential}) guarantee that all the cut nodes have potential 1. Constraint (\ref{milp:one0}) guarantees that there is at least one node that has potential 0. The existence of both potential 0 nodes and potential 1 nodes, excluding the cut nodes, implies that the remaining graph is disconnected. In the second case, $z = 1$, and the number of cut nodes is at least $|V| - 1$, guaranteed by constraint (\ref{milp:trivial}). Given that $M$ is sufficiently large (\eg, $M = 2|V|$), if $z = 0$, constraint (\ref{milp:trivial}) is satisfied; if $z = 1$, constraints (\ref{milp:one1}) and (\ref{milp:one0}) are satisfied. Thus, a node cut that satisfies either condition is a feasible solution of the following IP.

\begin{eqnarray}
\text{min} && \sum_{r \in C} c_r ~~~~~~~~~~~~~~~~~~~~~~~~~~~~~~~~~~~~~~~~~~~~~~~~~~\text{(IP)} \nonumber\\
\text{s.t.}
&& - y_i - y_j \leq p_i - p_{j} \leq y_i + y_j, ~~ \forall (i,j) \in E, \nonumber \\
&&  p_v \geq y_v, ~~ \forall v \in V, \label{milp:cutpotential}\\
&&  \sum_{v \in V} p_v -\sum_{v \in V} y_v - 1 \geq -Mz, \label{milp:one1} \\
&&  \sum_{v \in V} p_v - |V| + 1 \leq Mz, \label{milp:one0}\\
&&  \sum_{v \in V} y_v - |V| + 1 \geq -M(1 - z), \label{milp:trivial}\\
&& y_v \leq c_r, ~~ \forall r \in C, v \in \{v | r \text{ is the color of } v\}, \nonumber \\
&& c_r, p_v, y_v, z \in \{0,1\},  ~~ \forall v \in V, \forall r \in C. \nonumber
\end{eqnarray}

\subsection{A polynomially solvable case and an approximation algorithm} \label{sc:approx}
Although computing the minimum color node cut is NP-hard in general, there are special instances for which the value can be computed in polynomial time.
Let $V_i$ denote the nodes in $G$ that have color $i$. The \emph{induced graph} of $V_i$, denoted by $G[V_i]$, consists of $V_i$ and edges of $G$ that have both ends in $V_i$. We prove that if $G[V_i]$ is connected for all $i$, then the minimum color node cuts can be computed in polynomial time. It is worth noting that these special instances are reasonable representations for real-world interdependent networks, where a supply node is likely to support multiple directly connected nearby demand nodes.

Algorithm \ref{al:span1} computes the minimum color \st node cut in $G$ where $G[V_i]$ is connected $\forall i$, for any non-adjacent $(s,t)$ pair.
\begin{algorithm}[h]
\caption{Computation of the minimum color \st node cut in $G$ where $G[V_i]$ is connected $\forall i$.}
\label{al:span1}
\begin{enumerate}
\item Construct a new graph $G'$ from $G$ as follows. Contract the nodes $V_i$, which have the same color $i$, into a single node $u_i$. Connect $u_i$ and $u_j$ if and only if there is at least one edge between $V_i$ and $V_j$.
    Connect $s'$ to $\{u_i | s \text{ is connected to } V_i \}$, and connect $t'$ to $\{u_i | t \text{ is connected to } V_i \}$.
\item Compute the minimum $s't'$ node cut in $G'$, in which every node has a distinct color. The minimum color \st node cut in $G$ is given by the colors of the $s't'$ cut nodes in $G'$.
\end{enumerate}
\end{algorithm}

The following lemma proves the correctness of Algorithm~\ref{al:span1}.

\begin{lem} \label{th:span1}
 The $s't'$ node connectivity in $G'$ equals the value of the minimum color $st$ node cut in $G$, if $G[V_i]$ is connected,~$\forall i$, and $(s,t)$ are non-adjacent.
\end{lem}
\begin{proof}
  We aim to prove that there is a one-to-one mapping between a color \st node cut in $G$ and an $s't'$ node cut in $G'$, from which the result follows.

  One direction is simple. Let $C$ be the set of colors that appear in $G$. For any \st color node cut $C^{st}_c$ in $G$, after removing all (or a certain subset) of nodes with colors in $C^{st}_c$, there does not exist a sequence of colored nodes that connect $s$ and $t$. Two nodes $u_i, u_j$ are connected in $G'$ only if nodes with color $i$ and nodes with color $j$ are connected in $G$. Thus, there does not exist a sequence of nodes with colors in $C \setminus C^{st}_c$ that connect $s'$ and $t'$ in $G'$.

  To prove the other direction, consider any $s't'$ node cut in $G'$ and denote it by $V^{s't'}$. Let $V^{st} \subseteq G$ be a set of nodes with colors in $C_\text{color} = \{i | u_i \in V^{s't'} \}$. We aim to prove that $V^{st}$ is a superset of an \st node cut in $G$.

  If $V^{st}$ does not contain $s$ or $t$, after removing $V^{st}$ from $G$, no edge exists between the component that contains $s$ and the component that contains $t$. To see this, note that if no edge exists between $u_i$ and $u_j$ in $G'$, then no edge exists between any color $i$ node and any color $j$ node in $G$.

  If $V^{st}$ contains $s$, we need to prove that $V^{st} \setminus s$ is an \st cut in $G$. In Step 1 of Algorithm \ref{al:span1}, $s'$ is connected to all neighbors $N(s') := \{u_i | s \text{ is connected to } V_i \}$ in $G'$. After removing $V^{s't'}$, $N(s')$ are either removed or disconnected from $t'$. Therefore, the neighbors of $s$ in $G$ are either removed or disconnected from $t$ after removing $V^{st} \setminus s$.

  The same analysis proves that if $V^{st}$ contains $t$, then $V^{st} \setminus t$ is an \st cut in $G$. Similarly, if $V^{st}$ contains both $s$ and $t$, then $V^{st} \setminus s,t$ is an \st cut in $G$. This concludes the proof that $V^{st}$ is a superset of an \st node cut in~$G$.
\end{proof}

\begin{rmk}
A similar result exists in the computation of the minimum color \st \emph{edge} cut under the condition that all the edges that have the same color are connected \cite{coudert2007shared}. The difference in our problem is that the source or destination may have the same color as the nodes in a cut. Thus, to prove that a set of colors $C_c^{st}$ is a color cut, we need to prove that removing nodes, excluding $s$ and $t$, with colors $C_c^{st}$ disconnects $s$ and $t$. Thus, the proof has to take care of multiple corner cases.
\end{rmk}

To compute the global minimum color node cut of a colored graph, it is necessary to consider two different cases, resulting from the definition of a node cut that allows the remaining graph to be either disconnected or reduced to a single node. Algorithm \ref{al:span1global} computes the exact value of the global minimum color node cut of $G$ where $G[V_i]$ is connected $\forall i$.
\begin{algorithm}[h]
\caption{Computation of the global minimum color node cut of $G$ where $G[V_i]$ is connected $\forall i$.}
\label{al:span1global}
\begin{enumerate}
\item Compute minimum color \st node cut $C_c^{st}$ for all non-adjacent \st pairs in $G$ by Algorithm \ref{al:span1}. Let $C_c^1$ denote the minimum size $C_c^{st}$ over all \st pairs. (The cut $C_c^1$ is the minimum color node cut of $G$ that partitions $G$ into more than one component.)
\item Compute the minimum set of colors $C_c^2$ that cover at least $n-1$ out of the $n$ nodes in $G$. (\Ie, if there exists a color $i$ that is carried by one node, then $C_c^2$ include all the colors except color $i$. If there is no color that is carried by a single node, then $C_c^2$ include all the colors.)
\item The minimum color node cut of $G$ is given by the smaller of $C_c^1$ and $C_c^2$.
\end{enumerate}
\end{algorithm}

We remark that the global minimum color node cut of $G$ \emph{can not} be computed by first contracting nodes that have the same color and then computing the global minimum node cut in the new graph, even if $G[V_i]$ is connected $\forall i$. We only claim that the minimum color \st node cut in $G$ corresponds to the $s't'$ node cut in $G'$ obtained by Algorithm \ref{al:span1}, and that the global minimum color node cut of $G$ can be computed by Algorithm \ref{al:span1global}. Note that the topology of $G'$ depends on the choice of $s$ and $t$ (see Step 1 of Algorithm \ref{al:span1}).




The above result can be used to develop an approximation algorithm to compute the minimum color node cuts in an arbitrary colored graph where the induced graph $G[V_i]$ is not necessarily connected. To approximate the value of the minimum color \st node cut, the algorithm is a slight modification of Algorithm \ref{al:span1}. Instead of contracting $G[V_i]$ into a single node, in the new algorithm, \emph{each connected component} of $G[V_i]$ is contracted into a single node. Let the new graph be $G''$, and connect $s'', t''$ to the nodes contracted by the components in $G$ that are connected to $s,t$, respectively. The performance of the algorithm is given by Lemma \ref{th:approx}.
\begin{lem} \label{th:approx}
The $s''t''$ node connectivity in $G''$ is at most $q$ times the value of the minimum color \st node cut in $G$, where $q$ is the maximum number of components of $G[V_i]$, $\forall i$.
\end{lem}
\begin{proof}
Given that the induced graph $G[V_i]$ has at most $q$ components, after contracting each component into a node with color $i$, the number of nodes with color $i$ in $G''$ is at most $q$. Let $C^{st}_c$ denote a color node cut in $G$. By a similar reasoning as the proof of Lemma \ref{th:span1}, removing nodes with colors $C^{st}_c$ disconnects $s''$ from $t''$ in $G''$. Let $c^{st}_{\min}$ denote the value of the minimum color $st$ node cut $C^{st}_{c \min}$ in $G$. The number of nodes in $G''$ with colors $C^{st}_{c \min}$ is at most $c^{st}_{\min} q$. Moreover, $c^{st}_{\min} q$ is no smaller than the $s''t''$ node connectivity $k^{s''t''}$. Equivalently, $c^{st}_{\min}$ is at least $k^{s''t''} / q$.
\end{proof}

The global minimum color node cut of $G$ can be approximated to within factor $q$, by approximating the minimum color \st cuts for all non-adjacent \st pairs and taking the minimum size cut, and continuing Steps 2 and 3 of Algorithm \ref{al:span1global}. We conclude this section by summarizing the performance of the approximation algorithms.
\begin{thm}
  Given a colored graph $G(V,E,\CC)$, let $V_i$ be the set of nodes that have color $i$. If there are at most $q$ components in the induced graph $G[V_i]$, $\forall i$, then the values of the minimum color \st node cut and the global minimum color node cut can be approximated to within factor $q$ in $O(|V|^{0.5} |E| + |V|^2)$ and $O(|V|^{2.5} |E|)$ time, respectively. Note that if $q = 1$ the exact solutions are obtained.
\end{thm}
\begin{proof}
  The fact that the minimum color \st node cut can be approximated to within factor $q$ follows from Lemma \ref{th:approx}. The contraction of connected nodes that have the same color takes $O(|V|^2)$ time, by updating the adjacency matrix representation of $G$. Adding $s''$ and $t''$ to $G''$ takes $O(|V|)$ time, by increasing the numbers of rows and columns of the adjacency matrix by two and adding the new connections. Computing the minimum node $s''t''$ cut in $G''$ takes $O(|V|^{0.5} |E|)$ time \cite{even1975network}. The total time of approximating the minimum color \st node cut is $O(|V|^{0.5} |E| + |V|^2)$.

  The global minimum color node cut of $G$ is the minimum over 1) $C_c^1$: the minimum color node \st cut $\forall st$, and 2) $C_c^2$: the minimum number of colors that cover at least $n-1$ nodes. Since the value of the minimum color \st node cut can be approximated to within factor $q$, the minimum over all non-adjacent \st pairs, $|C_c^1|$, can also be approximated to within factor $q$. Moreover, the exact value of $|C_c^2|$ can be obtained in $O(|V|)$ time. Thus, the global minimum color node cut of $G$ can be approximated to within factor $q$. The number of non-adjacent \st pairs is at most $|V|^2/2$. The contraction of nodes with the same color can be computed once and reused. Computing the connections between $s'', t''$ and the contracted nodes takes $O(|V|)$ time for each $(s'', t'')$ pair. Computing the minimum node $s''t''$ cut in $G''$ takes $O(|V|^{0.5} |E|)$ time for each $(s'', t'')$ pair. Thus, the computation of $|C_c^1|$ requires $O(|V|^2 + |V|^{0.5} |E||V|^2 + |V||V|^2) = O(|V|^{2.5} |E|)$ time.

  We remark that although there are faster algorithms to compute the global minimum node cut (\eg, \cite{gabow2006using}), not all the accelerations can be applied to our problem. For example, computing $(k+1)|V|$ pairs of minimum \st node cut is enough to obtain the global minimum node cut in a graph $G$, where $k$ is the node connectivity of $G$, because at least one node among $k+1$ nodes does not belong to a minimum cut and can be a source or destination node. However, this does not hold in our problem, where the number of nodes covered by a minimum color node cut can be large, and the \st node connectivity for $\Theta(|V|^2)$ \st pairs should be evaluated.
\end{proof}

\section{Maximizing the supply node connectivity}
\label{sc:assign}
In this section, we develop supply-demand assignment algorithms to maximize the supply node connectivity of the demand network. Given a fixed demand network topology, the robustness of the demand network depends on the assignment of supply nodes for each demand node. For example, if every node in a cut depends on the same set of supply nodes, then removing these supply nodes could disconnect the demand network. In contrast, if different nodes in every cut depend on different supply nodes, then a larger number of supply nodes should be removed to disconnect the demand network.

For simplicity, in this section, we assume:
\begin{enumerate}
\item Every demand node has a fixed number of supply nodes, denoted by $n_s$.
\item Every supply node can support an arbitrary number of demand nodes.
\end{enumerate}
The total number of supply-demand pairs is $n_1 n_s$, where $n_1$ is the number of nodes in the demand network $G_1$. In Section \ref{sc:twoway}, we study the case where the number of nodes supported by every supply node is fixed as well, and study the interdependence assignment that maximizes the supply node connectivity of both $G_1$ and $G_2$.

The supply-demand assignment problem can be stated as follows in the context of a colored graph. Given a graph $G(V,E)$ and colors $C$, assign a color $c_i \in C$ to each node, such that the value of the minimum color node cut of $G$ (or the minimum color \st node cut for $s,t \in V$) is maximized. Graph $G$ is the transformed graph of the demand graph $G_1$, obtained by Algorithm \ref{al:tran}, where each node is replicated into $n_s$ nodes.

Under the first assumption, according to Corollary \ref{th:cutvalue}, the node connectivity of $G$ is $k = k_1 n_s$, where $k_1$ is the node connectivity of the demand graph $G_1$. Under any color assignment, the minimum color node cut of $G$ is at most $k$. Moreover, the minimum color node cut of $G$ is upper bounded by $n_c$, the total number of available colors (\ie, the total number of supply nodes in $G_2$). We aim to assign colors to nodes in order for the value of the minimum color node cut to be close to $\min(k, n_c)$. If the value of the minimum color node cut is $\min(k, n_c) / \alpha$ under an assignment algorithm $A$, then $A$ is an $\alpha$-approximation algorithm.

\subsection{Maximizing the \st supply node connectivity by path-based assignment}
We first propose Algorithm \ref{al:maxst} that maximizes the value of the minimum color \st node cut, which is simple but provides insight towards maximizing the value of the global minimum color node cut of a graph. 
\begin{algorithm}[h]
\caption{Path-based Color Assignment.}
\label{al:maxst}
\begin{enumerate}
\item Compute the \st node connectivity $k^{st}$. Identify $k^{st}$ node-disjoint \st paths.
\item Assign the same color to all the nodes in a path. If $n_c \geq k^{st}$, assign a distinct color to each path. If $n_c < k^{st}$, assign a distinct color to each of $n_c$ paths, and assign an arbitrary color to each remaining path.
\end{enumerate}
\end{algorithm}

For the $k^{st}$ node-disjoint \st paths, any pair of paths do not share the same color if there are sufficient colors ($n_c \geq k^{st}$), by the assignment in Algorithm \ref{al:maxst}. Thus, $s$ and $t$ stay connected after removing fewer than $k^{st}$ colors. On the other hand, if $n_c < k^{st}$, there exist $n_c$ paths with distinct colors, and $s$ and $t$ stay connected after removing fewer than $n_c$ colors. To summarize, the performance of Algorithm \ref{al:maxst} is given by the following theorem.
\begin{thm} \label{th:path}
  The value of the minimum color \st node cut is $\min(k^{st}, n_c)$ if the colors are assigned according to the \emph{Path-based Color Assignment} algorithm, where $n_c$ is the number of colors and $k^{st}$ is the \st node connectivity.
\end{thm}

It is worth noting that assigning the same color to multiple nodes in a path does not reduce the value of the minimum color \st node cut, compared with assigning a distinct color to each node. The reason is that, a path is disconnected as long as at least one node in the path is removed. To generalize, if a set of nodes together form a ``functional group'', it is better for nodes in the same group to share the same risk. In contrast, nodes in different groups should avoid sharing the same risk. We leverage this idea to maximize the global minimum color cut of a graph.

\subsection{Maximizing the global supply node connectivity by CDS-based assignment}
In the remainder of this section, we consider the color assignment that maximizes the global minimum color node cut of a graph. It is helpful to identify the group of nodes that support graph connectivity, analogous to nodes in a path that support \st connectivity. Indeed, nodes in a \emph{connected dominating set} (CDS) form such a group. A CDS is a set of nodes $S$ such that the induced graph $G[S]$ is connected and that every node in $V$ either belongs to $S$ or is adjacent to a node in $S$. If none of the nodes $S$ are removed, then the graph stays connected regardless of the number of removed nodes in $V \setminus S$. Namely, any subset of nodes $V \setminus S$ is not a node cut of the graph.

The natural analog of node-disjoint \st paths is (node) disjoint CDS, which support graph connectivity. The failures of nodes in one CDS do not affect another disjoint CDS, while a survived CDS suffices to keep the graph connected.
CDS partitions, which partition nodes of $G(V,E)$ into multiple disjoint CDS, have been studied in \cite{censor2014distributed, censor2014new, Censor-Hillel2015}. If the node connectivity of $G(V,E)$ is $k$ and $G(V,E)$ has $n$ nodes, then $\Omega(k / \log^2 n)$ node-disjoint CDS can be obtained in nearly linear time $O(m \text{ polylog } m)$, where $m$ is the number of edges \cite{censor2014distributed, Censor-Hillel2015}.

We propose Algorithm \ref{al:maxglobal} that assigns colors based on CDS partitions.
\begin{algorithm}[h]
\caption{CDS-based Color Assignment.}
\label{al:maxglobal}
\begin{enumerate}
\item Compute the node connectivity $k$ of $G$. Identify $k^\text{CDS} = \Omega(k / \log^2 n)$ node-disjoint CDS using the algorithm in~\cite{Censor-Hillel2015}.
\item Assign the same color to all the nodes in a CDS. If $n_c \geq k^\text{CDS}$, assign a distinct color to each CDS. If $n_c < k^\text{CDS}$, assign a distinct color to each of $n_c$ CDS, and assign an arbitrary color to each remaining CDS.
\end{enumerate}
\end{algorithm}

The performance of Algorithm \ref{al:maxglobal} can be analyzed in a similar approach to that of Algorithm \ref{al:maxst}. If $n_c \geq k^\text{CDS}$, each CDS has a distinct color, and the graph stays connected after removing fewer than $k^\text{CDS}$ colors. If $n_c < k^\text{CDS}$, $n_c$ CDS have distinct colors, and the graph stays connected after removing fewer than $n_c$ colors. Therefore, the value of the minimum color node cut is at least $\min(k^\text{CDS}, n_c)$. The performance of Algorithm \ref{al:maxglobal} is summarized by the following theorem.

\begin{thm}\label{th:cds}
The value of the minimum color node cut of $G$ is at least $\min(\Omega(k / \log^2 n), n_c)$ if the colors are assigned according to the \emph{CDS-based Color Assignment} algorithm, where $n_c$ is the number of colors, $n$ is the number of nodes, and $k$ is the node connectivity of $G$. The \emph{CDS-based Color Assignment} algorithm is an $O(\log^2 n)$-approximation algorithm.
\end{thm}

\subsection{Maximizing the global supply node connectivity by random assignment} \label{sc:random}
Finally, we study a Random Assignment algorithm. The algorithm is to assign each node a color randomly with equal probability. The intuition behind the Random Assignment algorithm is that nodes in a small cut are unlikely to be assigned with the same color if the number of colors is large. Thus, removing the nodes associated with a small number of colors is unlikely to disconnect the graph.

In fact, the Random Assignment algorithm has provably good performance. The analysis relies on the recently studied \emph{vertex sampling} problem in \cite{Censor-Hillel2015}. We first restate a sampling theorem in \cite{Censor-Hillel2015} as follows.

\begin{lem}[Theorem 6 in \cite{Censor-Hillel2015}]
\label{th:sampling}
Consider a graph $G$ in which each node is removed independently with a given probability $1-p$. For $ 0 < \delta < 1$, if the probability that a node is not removed satisfies $p \geq \beta \sqrt{\log (n/\delta) / k}$ for a sufficiently large constant $\beta$, then the remaining graph is connected with probability at least $1 - \delta$, where $n$ is the number of nodes and $k$ is the node connectivity of $G$.
\end{lem}

This sampling theorem provides a sufficient condition for a graph to be connected with high probability after its nodes are randomly removed. In particular, we use the following corollary.

\begin{crl}
\label{th:constremoval}
Given a graph $G$ with $n$ nodes and node connectivity $k = \omega(\log n)$, if each node is removed with up to a \emph{constant} probability $1-p < 1$, then the remaining nodes in $G$ are connected with probability $1 - \delta$ where $\delta = O(n e^{-\alpha k})$ for some constant $\alpha$.
\end{crl}
\begin{proof}
Given that the probability $p$ that each node remains in $G$ is at least a constant greater than zero, from Lemma \ref{th:sampling} we know that the probability $\delta$ that $G$ is disconnected satisfies the following equation.
\begin{align*}
k (p / {\beta})^2 &= \log(n/\delta), \\
\delta &= n e^{-\alpha k},
\end{align*}
where
$\alpha = (p / {\beta})^2$ is a constant.

Moreover, since $k = \omega(\log n)$, $\delta = n e^{-\alpha k} \leq n^{-1} = o(1).$ The probability that the remaining nodes are connected is high.
\end{proof}

On the other hand, if $k = O(\log n)$, $\beta \sqrt{\log(n/\delta) / k} \geq \beta \sqrt{\log (n) / k} = \Omega(1)$. The condition in Lemma \ref{th:sampling} cannot be satisfied, unless the hidden constant in $k = O(\log n)$ is large. Thus, the probability that the graph is disconnected after randomly removing a given fraction of nodes cannot be bounded using this approach. For simplicity, in the following we focus on graphs where $k = \omega(\log n)$.

In a colored graph $G$ where nodes are randomly colored using a total of $n_c$ colors, removing nodes with colors that belong to a given set of $k'$ colors is equivalent to removing each node with probability $k'/n_c$. The probability of removing a node is at most a constant, by restricting $k'$ to be at most $(1 - \epsilon)n_c$ for a constant $\epsilon > 0$. Thus, by Corollary \ref{th:constremoval}, the probability that $G$ is disconnected after removing nodes with a given set of $k'$ colors is small. By a union bound over ${n_c \choose k'}$ combinations of $k'$ colors, the probability $p_\text{union}$ that $G$ is disconnected after removing nodes with \emph{any} set of $k'$ colors can be bounded. If $p_\text{union}$ is small, and the remaining nodes form a CDS with high probability (such that removing any subset of nodes with any $k'$ colors does not disconnect $G$), then the value of the minimum color node cut of $G$ is at least $k' + 1$ with high probability. We next fill in the details of the proof,
and our approach closely follows the approach of computing node connectivity after random node sampling in~\cite{Censor-Hillel2015}.


\begin{thm}
\label{th:random}
By assigning a color uniformly at random to each of the $n$ nodes of $G$, the value of the minimum color node cut of $G$ is $\Theta (\min(k, n_c))$ with high probability, where $n_c$ is the number of colors and $k = \omega(\log n)$ is the node connectivity of $G$. If, in addition, $k = \omega(n_c)$, then the value of the minimum color node cut of $G$ is at least $(1 - \epsilon) n_c$ with high probability for any constant $\epsilon > 0$.
\end{thm}
\begin{proof}

We prove the theorem under three cases: i) $k = \Theta(n_c)$; ii) $k = \omega(n_c)$; and iii) $k = o(n_c)$. In all of the three cases, $k = \omega(\log n)$.

\textbf{i) First we consider the case where $k = \Theta (n_c)$.} For $k' \leq (1 - \epsilon)n_c$, where $\epsilon > 0$ is a constant, the probability that $G$ is disconnected after removing the nodes covered by a randomly selected set of $k'$ colors is $O(ne^{-\alpha k})$, for a constant $\alpha$ (Corollary \ref{th:constremoval}). The total number of $k'$ color combinations among the $n_c$ colors is ${n_c \choose k'} \leq (\frac{e n_c}{k'})^{k'}$. Thus, by the union bound, the probability that $G$ is disconnected after removing nodes with any $k'$ colors is at most $p_\text{union-1} = O(n e^{-\alpha k} (\frac{e n_c}{k'})^{k'})$. Let $k' = \alpha \min(k, n_c) / (2 \eta) \leq (1 - \epsilon)n_c$, where $\eta$ satisfies $\eta = \log \frac{e n_c}{k'} = \log \frac{2 \eta e n_c}{ \alpha \min(k, n_c)}$ and is a constant.
\begin{align*}
\log p_\text{union-1} &\leq \log (n e^{-\alpha k} (\frac{e n_c}{k'})^{k'}) \\
&= \log n - \alpha k + k' \log \frac{e n_c}{k'} \\
&= \log n - \alpha k + \alpha \min(k, n_c)/2 \\
&\leq \log n - \alpha k / 2 \\
&\leq - \gamma \log n,
\end{align*}
for a constant $\gamma > 0$. The last inequality follows from $k = \omega(\log n)$. Therefore, the probability that $G$ is disconnected is at most $n^{-\gamma} = o(1)$.

The above approach proves that with high probability, removing nodes with any $k'$ colors does not disconnect $G$. Before concluding that the value of the minimum color node cut of $G$ is at least $k'$, we need to prove that removing \emph{any subset} of nodes with any $k'$ colors does not disconnect $G$ (recall Definition \ref{def:colorglobal} of a color node cut). A sufficient condition is that the remaining nodes form a dominating set of $G$.

Since the node connectivity of $G$ is $k$, the minimum degree of a node in $G$ is at least $k$. The probability that all the neighbors of a node are removed is $(k'/n_c)^k$. Let $k' \leq (1 - \epsilon) n_c$ for a constant $\epsilon > 0$. The probability that there is at least one node whose neighbors are all removed can be upper bounded using the union bound
\begin{equation}\label{eq:domination}
  p_\text{union-2} = n (k'/n_c)^k \leq n (1 - \epsilon)^k = o(1).
\end{equation}

The last inequality follows from $k = \omega(\log n)$.
With probability $1 - o(1)$, there does not exist a node whose neighbors are all removed. Thus, the remaining nodes form a dominating set.

To conclude, with probability at least $1- p_\text{union-1} - p_\text{union-2} = 1 - o(1)$, the value of the minimum color node cut of $G$ is at least $k' = \Theta(k)$ if $k = \Theta(n_c)$ and $k = \omega(\log n)$.


\textbf{ii) Next we consider the case where $k = \omega (n_c)$.} 
Let $k' = (1 - \epsilon) n_c$.
\begin{align*}
\log p_\text{union-1} &\leq \log (n e^{-\alpha k} (\frac{e n_c}{k'})^{k'}) \\
&= \log n - \alpha k + k' \log \frac{e n_c}{k'} \\
&\leq \log n - \alpha k + 2k'
 \leq - \gamma \log n,
\end{align*}
for a constant $\gamma$. The last inequality holds because $k' = o(k)$ (equivalently, $(1 - \epsilon) n_c = o(k)$ and $k = \omega(n_c)$) and $\log n = o(k)$ (equivalently, $k = \omega(\log n)$). The value of the minimum color node cut of $G$ is at least $(1 - \epsilon) n_c$ with probability $1 - p_\text{union-1} - p_\text{union-2} = 1 - o(1)$.

\textbf{iii) Finally we consider the case where $k = o(n_c)$.} Directly using $p_\text{union-1}$ would incur an $O(\log n_c)$ gap from the optimal $k'$ (\ie, $k' = \Omega(k/\log n_c)$), because the number of $k'$ out of $n_c$ choices is large and the union bound $p_\text{union-1}$ is too weak. However, it is possible to reduce the number of choices, at the cost of removing a larger number of nodes. We use the same approach as in \cite{Censor-Hillel2015}. Partition the colors into $2k' = o(n_c)$ groups. Instead of removing nodes with colors in a selected set of $k'$ colors, we consider removing nodes with colors in a selected set of $k'$ color groups, which consists of around $n_c/2$ colors. The probability that each node is removed is $1/2$. The probability that $G$ become disconnected is still $\delta = O(n e^{-\alpha k})$. The total number of events (\ie, combinations of $k'$ color groups out of $2k'$ color groups) is reduced to ${2 k' \choose k'} \leq (2e)^{k'}$. For $k' = \alpha k / (2 \log(2e))$,
\begin{align*}
\log p_\text{union-3} &\leq \log (n e^{-\alpha k} (\frac{2 e k'}{k'})^{k'}) \\
&= \log n - \alpha k + k' \log(2e) \\
&\leq \log n - \alpha k/2
 \leq - \gamma \log n,
\end{align*}
for a constant $\gamma$.

Thus, the value of the minimum color node cut of $G$ is at least $k' = \alpha k / (2 \log(2e)) = \Theta(k)$ with high probability $1 - p_\text{union-3} - p_\text{union-2} = 1 - o(1)$.
\end{proof}

Theorem \ref{th:random} proves that the Random Assignment algorithm is an $O(1)$-approximation algorithm if $k = \omega(\log n)$. If $k = O(\log n)$, under any assignment the minimum color node cut value is at least one, and the approximation ratio is at most $O(\log n)$.

\section{Bidirectional interdependence}
\label{sc:twoway}
In the previous sections, we considered a one-way dependence model. In this section, we extend the results to a \emph{bidirectional interdependence} model.
Let $G_1(V_1, E_1)$ and $G_2(V_2, E_2)$ denote two interdependent networks. \emph{Interdependence edges} connect nodes between two networks, which represent their supply-demand relationship. The key difference from the one-way dependence model is that the interdependence edges are \emph{bidirectional} (\ie, if node $v \in G_1$ depends on node $u \in G_2$, then $u$ depends on $v$ as well).

If a node $v$ in $G_1$ fails due to the failures of its supply nodes in $G_2$, then the failure of $v$ does not lead to further node failures (due to a lack of supply) in $G_2$, because all the nodes in $G_2$ that depend on $v$ have failed. Otherwise, $v$ would not have failed in the first place. Therefore, the evaluation of supply node connectivity in the bidirectional interdependence model follows the same methods as the one-way dependence model. What remains to be developed is the \emph{interdependence assignment} that maximizes the supply node connectivity of \emph{both} networks.


We assume that there are $n_{si}$ interdependence edges adjacent to each of the $n_i$ nodes in $G_i$ ($\forall i = \{1,2\}$). The total number of bidirectional interdependence edges is $n_1 n_{s1} = n_2 n_{s2}$. Under this assumption, a node in $G_i$ is functional if at least one of its adjacent $n_{si}$ interdependence edges is connected to a remaining node (\ie, a node that has not been removed) in $G_j$ ($\forall i,j = \{1,2\}, i \neq j$).

We now give an overview of the bidirectional interdependence assignment algorithms. To extend the CDS-based color assignment to interdependence assignment, we aim to avoid disjoint CDS sharing the same supply nodes as much as possible, in both networks. Nodes in $G_1$ are partitioned into groups of size $n_{s2}$, and nodes in $G_2$ are partitioned into groups of size $n_{s1}$. Interdependence is assigned between each group in $G_1$ and each group in $G_2$. Consider a group $P_1 \in G_1$, and a corresponding group $P_2 \in G_2$. Every node $v_1 \in P_1$ depends on all the nodes in $P_2$, and every node $v_2 \in P_2$ depends on all the nodes in $P_1$. The key is to partition nodes in $G_1$ and $G_2$ into groups. 
The partition is obvious when the number of nodes in each CDS in $G_i$ is a multiple of $n_{sj}$ ($\forall i,j = \{1,2\}, i \neq j$), in which case disjoint CDS do not share any supply node. See Fig. \ref{fig:cds} for an illustration. Otherwise, in general, disjoint CDS may have to share some supply nodes. As we will prove later, the supply node connectivity will be reduced by at most a half, compared with the ideal case where disjoint CDS do not share any supply node. The same analysis applies to the path-based assignment that maximizes the \st supply node connectivity, and is omitted.

\begin{figure}[h]
\begin{centering}
\leavevmode\includegraphics[width=.9\linewidth]{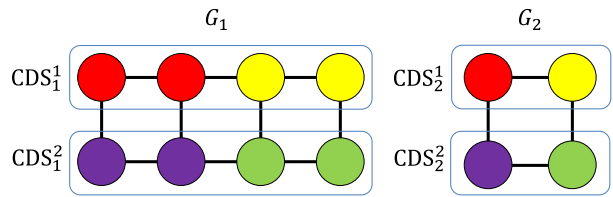}
\caption{An example of the partition of CDS nodes into groups. Every node in $G_1$ and $G_2$ has $n_{s1} = 1$ and $n_{s2} = 2$ supply nodes, respectively. Each CDS in $G_i$ is partitioned into two groups of size $n_{sj}$ ($i, j \in \{1,2\}, i \neq j$). In each graph, nodes that have the same color are in the same group. Between two graphs, nodes in groups with the same color are interdependent. The partition achieves the optimal supply node connectivity: 2 and 4 for $G_1$ and $G_2$, respectively.}
\label{fig:cds}
\end{centering}
\end{figure}

\subsection{CDS-based interdependence assignment}
We develop an algorithm to partition the nodes in $G_1$ into groups of size $n_{s2}$, and to partition the nodes in $G_2$ into groups of size $n_{s1}$. A group of size $n_{sj}$ is \emph{empty} if it contains no node, is \emph{full} if it contains $n_{sj}$ nodes, and is \emph{occupied} if it contains more than zero but fewer than $n_{sj}$ nodes. If $|V_i| / n_{sj}$ is an integer, we aim to partition $V_i$ into $|V_i| / n_{sj}$ full groups. Otherwise, if $|V_i| / n_{sj}$ is not an integer, we aim to partition $V_i$ into $\lfloor |V_i| / n_{sj} \rfloor$ full groups and one occupied group that contains $|V_i^{*}| = |V_i| - \lfloor {|V_i| / n_{sj}} \rfloor n_{sj} $ nodes ($\forall i, j \in \{1,2\}, i \neq j$), where $V_i^{*}$ denotes the nodes in the occupied group.
Since $|V_1|/n_{s_2} = |V_2|/n_{s_1}$, the total number of groups are the same in both $G_1$ and $G_2$.

Interdependence is assigned between nodes in two groups, one from each graph. For each node in $V_i \setminus V_i^{*}$, there are $n_{si}$ supply nodes. For each node in $V_i^{*}$, there are $|V_j^{*}| < n_{si}$ supply nodes. (Multiple interdependence edges exist between some nodes in $V_i^{*}$ and some nodes in $V_j^{*}$).
Given that nodes within a group depend on the same set of supply nodes while different groups of nodes depend on different supply nodes, we aim to partition nodes into groups such that \emph{a large number of groups} need to be removed in order to disconnect all the CDS. Consequently, a large number of supply nodes need to be removed in order to disconnect all the CDS. The partition of $V_i$ into $\lfloor {|V_i| / n_{sj}} \rfloor$ full groups follows Algorithm \ref{al:partitionCDS}. The remaining nodes (if any) form an occupied group $V_i^*$ if $|V_i| / n_{sj}$ is not an integer.

We denote by $h$ the number of disjoint CDS in $G_i$. Using the algorithm in \cite{Censor-Hillel2015}, $h = \Omega(k_i / \log^2 n_i)$ disjoint CDS can be computed. If there are extra nodes in $V_i$ that do not belong to the $h$ CDS, then these nodes are added to the largest CDS. Note that adding extra nodes to a CDS still yields a CDS, since these nodes are adjacent to the nodes in the original CDS.

\begin{algorithm}[h]
\caption{Assign nodes $V_i$ into $\lfloor {|V_i| / n_{sj}} \rfloor$ full groups of size $n_{sj}$.}
\label{al:partitionCDS}
\begin{enumerate}
\item Sort the $h$ disjoint CDS in the ascending order of their sizes. Denote the nodes in the $l$-th CDS in $G_i$ by $N^l$, $l = 1,2,\dots,h$. 
\item For $l$ from 1 to $h$, start with an empty group if available, and assign nodes from $N^l$ into the group. Repeat until all nodes are assigned. If there are not enough empty groups, assign the rest nodes into occupied groups until these groups become full. 
\item The algorithm terminates when the $\lfloor {|V_i| / n_{sj}} \rfloor $ groups become full.
\end{enumerate}
\end{algorithm}

The following example illustrates Step 2 of the algorithm. Before assigning $N^l$, there are \emph{enough} empty groups if the number of empty groups is at least $\lceil |N^l| / n_{sj} \rceil$. Nodes $N^l$ are assigned to $\lfloor |N^l| / n_{sj} \rfloor$ groups, which then become full. If $|N^l|/n_{sj}$ is not an integer, the remaining $|N^l| - \lfloor |N^l| / n_{sj} \rfloor n_{sj} \leq n_{sj} - 1$ nodes are assigned to another empty group and the group becomes occupied. On the other hand, if there are $n_r < \lceil |N^l| / n_{sj} \rceil$ empty groups before assigning $N^l$, then $n_r n_{sj}$ nodes in $N^l$ are assigned to the $n_r$ groups. The remaining nodes in $N^l$ and nodes in $N^{l+1},\dots, N^{h}$ are assigned to the already occupied groups.

The algorithm is further illustrated by Fig. \ref{fig:partition}. Suppose that $G_1$ has 12 nodes, and has three disjoint CDS, consisting of $|N^1| = 2, |N^2| = 4, |N^3| = 6$ nodes, respectively, and that $n_{s2} = 3$. Our goal is to assign the 12 nodes in $G_1$ to 4 groups of size 3. Before assigning nodes in $N^1$, all the four groups are empty. Thus, the two nodes in $N^1$ can be assigned to an empty group. After the assignment, the group becomes occupied, illustrated by the left figure in Fig. \ref{fig:partition}. Before assigning nodes in $N^2$, there are three empty groups. The assignment of $N^2$ uses two groups, one of which becomes full and the other becomes occupied (groups 2 and 3 in Fig. \ref{fig:partition}). Finally, when assigning $N^3$, there is only one empty group, and thus there are not enough empty groups to hold all the nodes in $N^3$. The last empty group can be assigned with 3 nodes. The remaining 3 nodes in $N^3$ are assigned to the occupied groups (\ie, groups 1 and 3 in Fig. \ref{fig:partition}).

\begin{figure}[h]
\begin{centering}
\leavevmode\includegraphics[width=.95\linewidth]{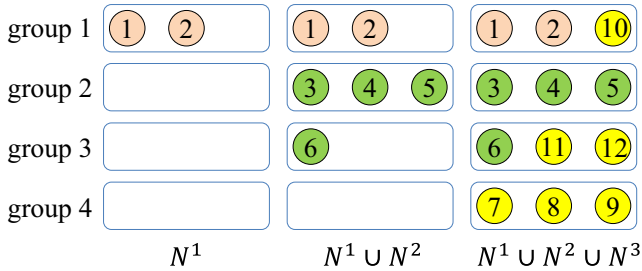}
\caption{Partition the CDS nodes $\{N^1$, $N^2$, $N^3\}$ into four groups of size three. The left, middle, right figures represent the snapshots after assigning nodes in $N^1, N^2, N^3$ in Step 2 of Algorithm \ref{al:partitionCDS}, respectively.}
\label{fig:partition}
\end{centering}
\end{figure}

We prove that disjoint CDS are sufficiently group-disjoint, by characterizing the number of groups that need to be removed to disconnect all the CDS.

\begin{lem}
\label{th:partition}
Let $V_i$ be assigned to groups according to Algorithm \ref{al:partitionCDS}. The minimum number of full groups that need to be removed, in order for each CDS to contain at least one removed node, is at least $\min(\lceil (h - 1)/2 \rceil, \lfloor |V_i|/n_{sj} \rfloor)$.
\end{lem}
\begin{proof}
  Let $|N^l|$ denote the number of nodes in the $l$-th CDS of $G_i$, $\forall l \in \{ 1,\dots, h \}$. If $|N^l|$ is a multiple of $n_{sj}$, $\forall l \in \{1,\dots, h\}$, then nodes in $N^{l_1}$ are assigned to different groups from nodes in $N^{l_2}$, $\forall l_1, l_2 \in \{ 1,\dots, h \}, l_1 \neq l_2$. To remove at least one node from each of the CDS, $h$ full groups need to be removed. 
  In the rest of the proof, we focus on the case where $|N^l|$ is not a multiple of $n_{sj}$ for some $l \in \{1,\dots,h\}$. 

  In the first few assignments in Algorithm \ref{al:partitionCDS} when there are enough empty groups, nodes in $N^{l_1}$ are assigned to different groups from nodes in $N^{l_2}$, $\forall l_1, l_2 \in \{ 1,\dots, k_\text{th}\}, l_1 \neq l_2$. 
  In order to disconnect all the CDS, at least one node should be removed from each CDS. The removed nodes in CDS $N^l, l = 1,\dots, k_\text{th}$ belong to at least $k_\text{th}$ distinct groups. Therefore, at least $k_\text{th}$ groups need to be removed in order to disconnect all the CDS. (Note that these groups become full by the end of Algorithm \ref{al:partitionCDS}.) 

  \emph{Determining $k_\text{th}$:}
  Consider one CDS $N^l$ ($l \in \{1,\dots,k_\text{th}\}$). If $|N^l|/n_{sj}$ is not an integer, one group occupied by $N^l$ is not full, and the group can still be assigned with $r^l \leq n_{sj} - 1$ extra nodes. If $|N^l|/n_{sj}$ is an integer, then $r^l = 0$. The total number of extra nodes that can be assigned into these occupied groups is $\sum_{l=1}^{k_\text{th}} r^l \leq k_\text{th}(n_{sj} - 1)$.

  Consider the assignment when there are not enough empty groups to hold all the nodes in $N^l, \forall l = k_\text{th}+1,\dots, h$. 

  1) If $|N^{k_\text{th}+1}| \leq n_{sj}$, then $|N^l| \leq n_{sj}, l = 1,\dots, k_\text{th}$. (Recall that the CDS are sorted in the ascending order of their sizes.) Nodes in each CDS $N^l$ belong to a single occupied group, $l = 1,\dots, k_\text{th}$. Moreover, since there is no empty group available when assigning nodes in $N^{k_\text{th}+1}$, all the empty groups have been used, and $k_\text{th} = \lfloor |V_i|/n_{sj} \rfloor$.

  2) If $|N^{k_\text{th}+1}| \geq n_{sj} + 1$, the number of remaining CDS is at most
  \begin{align*}
  h - k_\text{th} & \leq \lfloor \frac{\sum_{l=1}^{k_\text{th}} r^l + n_{sj} - 1} {n_{sj}+1} \rfloor + 1 \\
  &\leq \lfloor \frac{(k_\text{th} + 1)(n_{sj} - 1)}{n_{sj} + 1} \rfloor + 1
  \leq k_\text{th} + 1.
  \end{align*}
  To see this, note that there is no empty group available when assigning nodes in $\cup_{l = k_\text{th} + 2}^{h} N^l$. Otherwise, all the nodes in $N^{k_\text{th} + 1}$ would have been assigned to empty groups, which contradicts the assumption. Let $N^o \subseteq \cup_{l = k_\text{th} + 2}^{h} N^l$ denote the nodes that will be assigned to the occupied groups (occupied by nodes in $N^l , l \in \{1,\dots,k_\text{th}\}$). Let $N^{*} \subseteq \cup_{l = k_\text{th} + 2}^{h} N^l$ denote the remaining nodes that cannot be assigned to the $\lfloor |V_i|/n_{sj}\rfloor$ groups when $|V_i|/n_{sj}$ is not an integer. By definition, $N^o \cup N^{*} = \cup_{l = k_\text{th} + 2}^{h} N^l$. We know that $|N^o|$ is at most $\sum_{l=1}^{k_\text{th}} r^l$, which is the number of extra nodes that the occupied groups can fit. Moreover, $|N^{*}|$ is at most $|V_i| - \lfloor |V_i|/n_{sj}\rfloor n_{sj} \leq n_{sj} - 1$. Therefore, $\sum_{l=1}^{k_\text{th}} r^l + n_{sj} - 1$ is an upper bound on the number of nodes in $\cup_{l = k_\text{th} + 2}^{h} N^l$. Since the size of $N^l$ ($k_\text{th} + 2 \leq l \leq h$) is at least $n_{sj} + 1$, the first term in the summation is an upper bound on the number of CDS $N^{k_\text{th} + 2},\dots,N^h$. The additional one (second term in the summation) accounts for the CDS $N^{k_\text{th} + 1}$.

  In summary, given that the total number of CDS $h = k_\text{th} + (h - k_\text{th}) \leq k_\text{th} + (k_\text{th} + 1)$, we obtain $k_\text{th} \geq (h - 1)/2$. Since $k_\text{th}$ is an integer, $k_\text{th}$ is at least $\lceil (h - 1)/2 \rceil$.
\end{proof}

Given that $n_{si}$ supply nodes need to be removed in order to remove a full group of nodes of $V_i$, we have the following result.
\begin{thm}
  Given $G_i$ with $n_i$ nodes and node connectivity $k_i$, and that every node has $n_{si}$ supply nodes, $\forall i \in \{1,2\}$, assign interdependence between nodes in $G_1$ and the nodes in $G_2$ by groups, obtained in Algorithm \ref{al:partitionCDS}. Then, the supply node connectivity of $G_i$ is $\Omega(\min (k_i n_{si}/\log ^2 n_i, n_j))$, $\forall i,j \in \{1,2\}, i \neq j$.
\end{thm}

\begin{proof}
  Using the algorithm in \cite{Censor-Hillel2015}, $h = \Omega (k_i / \log^2 n_i)$ disjoint CDS can be found in $G_i$. By Lemma \ref{th:partition}, the number of full groups that should be removed in order to remove at least one node from each CDS is $\min (\lceil (h - 1)/2 \rceil, \lfloor n_i / n_{sj} \rfloor)$, $\forall i, j\in \{1,2\}, i \neq j$. 

  Each group of $V_i$ can be removed by removing $n_{si}$ supply nodes in $G_j$. Noting that $h = \Omega (k_i / \log^2 n_i)$ and that $n_i n_{si}/ n_{sj} = n_j$, the supply node connectivity of $G_i$ is $\Omega(\min (k_i n_{si}/\log ^2 n_i, n_j))$, $\forall i, j\in \{1,2\}, i \neq j$.
\end{proof}

We have proved that the CDS-based interdependence assignment algorithm is an $O(\log^2 n_i)$-approximation algorithm in maximizing the supply node connectivity of $G_i$, $\forall i \in \{1,2\}$. 




\subsection{Random interdependence assignment}
We study the random assignment in order to maximize the supply node connectivity of both graphs. The random assignment algorithm is to randomly match $n_{s1}$ copies of nodes in $G_1$ with $n_{s2}$ copies of nodes in $G_2$, and assign interdependence between matched nodes. Under the assignment, each of the $n_i$ nodes in $G_i$ is supported by $n_{si}$ nodes in $G_j$ ($i,j \in \{1,2\}, i \neq j$). 

The key difference of the analysis from the random assignment algorithm for the one-way dependence model is as follows. By randomly removing $k'$ nodes in $G_2$, $k'n_{s2}$ nodes in the transformed graph of $G_1$ (by Algorithm \ref{al:tran}) are removed. In contrast, in the one-way dependence model (Section \ref{sc:random}), every node is removed with probability $k'n_{s2}/n_1 n_{s1}$, and the total number of node removals follows a binomial distribution with mean $k'n_{s2}$.
We derive the following lemma that bounds the probability of a graph being disconnected after a constant fraction of nodes are removed, instead of each node being removed with a constant probability as in Corollary \ref{th:constremoval}.
\begin{lem}
\label{th:sampling2}
Given graph $G$ with $n$ nodes and node connectivity $k = \omega(\log n)$, after randomly removing up to a \emph{constant} (less than one) fraction of $n$ nodes, the remaining nodes in $G$ are connected with probability $1 - \delta$ where $\delta = O(n e^{-\alpha' k})$ for some constant $\alpha'$.
\end{lem}
\begin{proof}
We prove a stronger result that the remaining nodes form a connected dominating set (CDS) with high probability. In particular, we prove for the case where $(1 - \epsilon)(1 - p)n$ nodes are randomly removed, for a constant $\epsilon < 1$ and a constant $p \in (0,1)$.

Let $A(n_\text{rm})$ denote the event that the remaining nodes in $G$ form a CDS after randomly removing $n_\text{rm}$ nodes, where $n_\text{rm}$ is a deterministic value. Since adding extra nodes to a CDS still yields a CDS, $\Pr(A(n_\text{rm}))$ is decreasing in $n_\text{rm}$.

Consider the case where each node is randomly removed with probability $1 - p \in (0, 1)$. The number of removed nodes, $N_\text{rm}$, follows a binomial distribution with mean $(1-p)n$. Using the Chernoff bound, for a constant $\epsilon < 1$,
$$ \Pr(N_\text{rm} < (1 - \epsilon) (1 - p)n) \leq e^{-(1 - p)n \epsilon^2 / 2}. $$
The probability that the remaining nodes in $G$ form a CDS after removing $N_\text{rm}$ nodes is:
\begin{align}
  \Pr&(A(N_\text{rm})) =  \sum_{n_\text{rm} = 0}^{n} \Pr(A(n_\text{rm})) \Pr(N_\text{rm} = n_\text{rm}) \label{eq:total}\\
  \leq &  \Pr(A((1 - \epsilon) (1 - p)n)) \Pr(N_\text{rm} \geq (1 - \epsilon) (1 - p)n) \nonumber \\
  ~ & + 1 \Pr(N_\text{rm} < (1 - \epsilon) (1 - p)n) \label{eq:mono1}\\
  \leq & \Pr(A((1 - \epsilon) (1 - p)n)) + e^{-(1 - p)n \epsilon^2 / 2}, \label{eq:mono2}
\end{align}
where Eq. (\ref{eq:total}) follows from the law of total probability, Eq.~(\ref{eq:mono1}) follows from that $\Pr(A(n_\text{rm}))$ is non-increasing in $n_\text{rm}$, and Eq.~(\ref{eq:mono2}) follows from the Chernoff bound.
Thus,
$$ \Pr(A((1 - \epsilon) (1 - p)n)) \geq \Pr(A(N_\text{rm})) - e^{-(1 - p)n \epsilon^2 / 2}. $$

From the proof of Corollary \ref{th:constremoval}, we know that by removing $N_\text{rm}$ nodes, $G$ is disconnected with probability at most $ne^{-\alpha k}$, where $\alpha$ is a constant. Moreover, let $k'/n_c = 1 - p$ in Eq. (\ref{eq:domination}), the probability that the remaining nodes in $G$ do not form a dominating set is at most $n(1-p)^k$. Thus, by the union bound, the probability that the remaining nodes in $G$ do not form a connected dominating set is at most $1 - \Pr(A(N_\text{rm})) \leq ne^{-\alpha k} + n(1-p)^k$.

We now bound the probability that the remaining nodes in $G$ form a CDS, after randomly removing $p'n$ nodes, where $p' = (1 - \epsilon)(1 - p)$ is a constant.
\begin{align*}
  \Pr(A((1 - \epsilon) & (1 - p)n)) \geq  \Pr(A(N_\text{rm})) - e^{-(1 - p)n \epsilon^2 / 2} \\
  \geq & 1 - ne^{-\alpha k} - n(1 - p)^k - e^{-(1 - p)n \epsilon^2 / 2}
\end{align*}

Let $\alpha' = \min(\alpha, -\log(1 - p), (1-p)n \epsilon^2 /2k)$. Then $ne^{-\alpha k}, n(1-p)^k, e^{-(1 - p)n \epsilon^2 / 2} \leq ne^{-\alpha' k}$. Therefore, $\Pr(A(p'n)) \geq 1 - O(ne^{-\alpha' k})$. Moreover, since $\alpha$, $p$ are constants and $n = \Omega(k)$, $\alpha'$ is a constant.
\end{proof}

Then, following the analysis in Theorem \ref{th:random}, and noting Corollary \ref{th:cutvalue}, we obtain the following result.
\begin{thm}
  Given $G_i$ with $n_i$ nodes and node connectivity $k_i$, if each node in $G_i$ has $n_{si}$ supply nodes, by randomly matching $n_{si}$ copies of nodes in $G_i$ to $n_{sj}$ copies of nodes in $G_j$, and assigning interdependence between each pair of matched nodes, then the supply node connectivity of $G_i$ is $\Theta(\min(k_i n_{si}, n_j))$ with high probability, if $k_i n_{si} = \omega(\log (n_i n_{si}))$, $\forall i,j \in \{1,2\}, i \neq j$. If, in addition, $k_i n_{si} = \omega(n_j)$, then the supply node connectivity of $G_i$ is at least $(1 - \epsilon) n_j$ with high probability for any constant $\epsilon > 0$.
\end{thm}
\begin{proof}
  By Corollary \ref{th:cutvalue}, the transformed graph (by Algorithm \ref{al:tran}) $\tilde G_i$ has $n_i n_{si}$ nodes and node connectivity $k_i n_{si}$, $\forall i \in \{1,2\}$. The number of colors in $G_i$ is the number of nodes $n_j$ in $G_j$, $\forall i,j \in \{1,2\}, i \neq j$. Given Lemma \ref{th:sampling2}, the supply node connectivity of $G_i$ can be computed in the same approach as the proof for Theorem \ref{th:random}.
\end{proof}

Thus, the random assignment is an $O(1)$-approximation algorithm in maximizing the supply node connectivity of both $G_1$ and $G_2$, if $k_i n_{si} = \omega(\log (n_i n_{si}))$, $\forall i \in \{1,2\}$. If $k_i n_{si} = O(\log (n_i n_{si}))$, the approximation ratio is at most $O(\log (n_i n_{si}))$, since the supply node connectivity is at least one under any assignment, $\forall i \in \{1,2\}$.

\section{Numerical results}
\label{sc:simu}
In this section, we apply the algorithms in the previous sections and provide numerical results. We use MATLAB to generate network topologies and dependence assignment, and use JuMP \cite{LubinDunningIJOC} to compute the supply node connectivity by calling CPLEX to solve the integer programs in a workstation that has an Intel Xeon Processor (E5-2687W v3) and 64GB RAM. 

The key observations are as follows. First, the supply node connectivity for a network of reasonable size can be computed using the integer program in a short time. For example, the results can be obtained within one minute, for a network that has around 180 nodes and 650 edges. Second, the assignment algorithms have good performance even when the value of supply node connectivity is moderate. This complements the theoretical results that the assignment algorithms are optimal up to a constant or polylogarithmic factor. The numerical results therefore suggest that the algorithms are practical in the design of interdependent networks.

\subsection{\st supply node connectivity}
We use the XO communication network \cite{xo} of 60 nodes as an example of the demand network, and randomly generate 36 supply nodes (marked as triangles in Fig. \ref{fig:xo}) within the continental US. 

\begin{figure}[h]
\begin{centering}
\leavevmode\includegraphics[width=\linewidth]{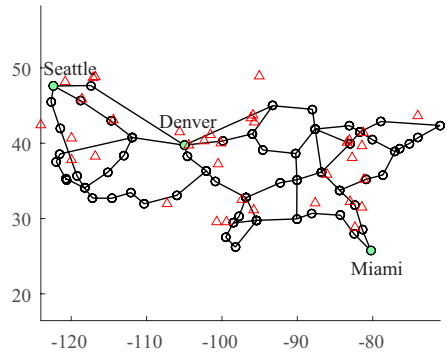}
\caption{XO network as a demand network, with randomly generated supply nodes. The $x$-axis represents longitude degrees (west), and the $y$-axis represents latitude degrees (north).}
\label{fig:xo}
\end{centering}
\end{figure}

Let each node in the XO network be supported by three nearest supply nodes. After transforming the network into a colored graph by Algorithm \ref{al:tran} and solving the MILP, we obtain that the supply node connectivity of the \st pair Seattle-Denver is 5. In contrast, the maximum \st supply node connectivity is 9, by assigning distinct supply nodes to each of the three node-disjoint paths (\ie, the path-based assignment outlined in Algorithm \ref{al:maxst}). As another example, the supply node connectivity of the \st pair Seattle-Miami is only 3, because one node in an \st path has the same set of three supply nodes as another node in a disjoint \st path. By assigning distinct supply nodes to two disjoint paths (Algorithm \ref{al:maxst}), the supply node connectivity of Seattle-Miami can be increased to 6.

\subsection{Global supply node connectivity}
If each node in the XO network is supported by its three nearest supply nodes, the global supply node connectivity is 3. In contrast, if each node is supported by three randomly chosen supply nodes, the global supply node connectivity can be increased to 5. It is close to the maximum possible global supply node connectivity 6, given that the node connectivity of the XO network is two and each node has three supply nodes. However, the CDS-based assignment (Algorithm \ref{al:maxglobal}) only guarantees that the supply node connectivity is at least 3, since there do not exist two disjoint CDS in the XO network.

\subsection{Bidirectional interdependence assignment}
We implement the bidirectional interdependence assignment algorithms on randomly generated Erdos-Renyi graphs. Let $G_i$ be an Erdos-Renyi graph with $n_i$ nodes. Let the probability that an edge exists between any two nodes be $p_i$. Each node in $G_i$ has $n_{si}$ supply nodes from $G_j$. Let $k_i$ denote the node connectivity of $G_i$. Recall that the maximum supply node connectivity is $k^s_{i \max} = \min(k_i n_{si}, n_j)$ ($i,j \in\{1,2\}, i \neq j$). Table \ref{fig:global2} depicts the supply node connectivity $k_i^s$ of $G_i$ under the CDS-based and random interdependence assignment algorithms. To obtain the numerical results for CDS-based interdependence assignment algorithm, instead of using the CDS partition algorithm in \cite{Censor-Hillel2015}, we use a greedy approach to compute the disjoint CDS, which has good performance for Erdos-Renyi graphs. The results are averaged over 10 instances for each of the two combinations of interdependent networks: 1) $n_1 = 50$, $n_2 = 75$, $p_1 = p_2 = 0.1$; 2) $n_1 = 50$, $n_2 = 75$, $p_1 = p_2 = 0.2$. From the results, we observe that the (near-linear time) CDS-based and the (linear time) random interdependence assignment algorithms yields near-optimal supply node connectivity in both graphs.

\begin{table}[h]
\centering
\caption{Supply node connectivity $k_i^s$ of random graphs under CDS-based and random assignments.}
\label{fig:global2}
\begin{tabular}{|l|l|l|l|l|l|l|}
\hline
$n_1$ & $p_1$ & $k_1$ & $n_{s1}$ & $k^s_{1 \max}$ & $k^s_1$ CDS & $k^s_1$ random \\ \hline
50    & 0.1   & 1.6   & 3        & 4.8             & 4.8        & 4.7          \\
50    & 0.2   & 3.6  & 3        & 10.8            & 10.2       & 10.0            \\ \hline
$n_2$ & $p_2$ & $k_2$ & $n_{s2}$ & $k^s_{2 \max}$ & $k^s_2$ CDS & $k^s_2$ random \\ \hline
75    & 0.1     & 2.4   & 2        & 4.8           & 4.6        & 4.6           \\
75    & 0.2    & 7.0     & 2        & 14.0            & 12.4       & 12.2 \\ \hline
\end{tabular}
\end{table}

\section{Conclusion}
\label{sc:conclude}
We studied the robustness of interdependent networks based on a finite-size, arbitrary-topology graph model. We defined supply node connectivity as a robustness metric, by generalizing the node connectivity in a single network. We developed integer programs to compute the supply node connectivity both for an \st pair and for a network, and developed approximation algorithms for faster computation. Moreover, we develop interdependence assignment algorithms to design robust interdependent networks.


Our study extends the shared risk group model, by considering that multiple risks together lead to the failure of a node. The color assignment algorithms in Section \ref{sc:assign} can be used as solutions to the less intensively studied design problems for the shared risk group model, to maximize the number of risks that a network can tolerate. 

\section*{Appendix}
\begin{proof}[Proof of Theorem \ref{th:np}]
  The minimum color node cut problem can be reduced from the \emph{vertex cover} problem. Given a graph $G'(V',E')$, the minimum vertex cover problem aims to select the minimum number of nodes $V^* \subseteq V'$ such that every edge in $E'$ is incident to at least one node in $V^*$.

  We construct a colored graph $G$ in which the value of the minimum color node cut equals the size the the minimum vertex cover in $G'$. Let $m'$ denote the number of edges in $G'$. Without loss of generality we assume that $m'$ is even. (Otherwise, one edge can be added parallel to any existing edge, which does not change the size of the minimum vertex cover.) Graph $G$ consists of four cliques of size $m'$ each. Nodes in every clique are divided into two disjoint sets of size $m'/2$. Four cliques are joined into a ring, by matching two disjoint set of $m'/2$ nodes of a clique to $m'/2$ nodes in each of the two adjacent cliques (see Fig. \ref{NPNodeCutG}).

  Then, assign colors to nodes in $G$. Consider two matchings $M_1$ and $M_2$ that connect two pairs of cliques. There are $m'$ edges in the union of the two matchings. Each edge $(v_1,v_2)$ in this union corresponds to an edge $(v'_1,v'_2)$ in $G'$. Let each node in $G'$ have a distinct color, and assign $v_1$ ($v_2$) the same color as $v'_1$ ($v'_2$). 
  Finally, assign each remaining node in $G$ a distinct color, and these remaining nodes are not adjacent to matchings edges $M_1$ or $M_2$. See Fig. \ref{NPNodeCutG} for an example of $G$, where $m' = 8$ and the number on each node represents its color, and the corresponding $G'$ represented by Fig. \ref{NPNodeCutGG}.

\begin{figure}[h]
\begin{centering}
\leavevmode\includegraphics[width=0.8\linewidth]{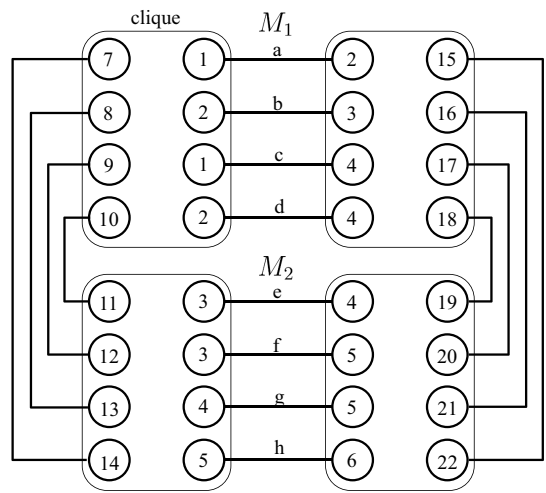}
\caption{In a colored graph $G$ where the number on each node represents its color, the minimum color node cut is $\{2,4,5\}$.}
\label{NPNodeCutG}
\end{centering}
\end{figure}

\begin{figure}[h]
\begin{centering}
\leavevmode\includegraphics[width=0.6\linewidth]{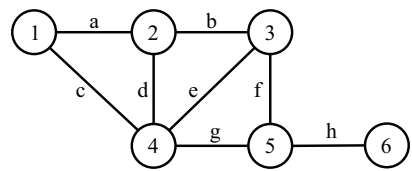}
\caption{The minimum vertex cover in $G'$ is $\{2,4,5\}$. }
\label{NPNodeCutGG}
\end{centering}
\end{figure}

  By removing at least one node incident to each of the $m'$ matching edges, $G$ becomes disconnected. In particular, the minimum color node cut of $G$ consists of a set of colors $C_c$ such that all the matching edges $M_1$ and $M_2$ are incident to at least one node that has a color in $C_c$. The nodes in $G'$ that have colors $C_c$ form the minimum vertex cover in $G'$, since every edge in $G'$ is adjacent to at least one node that has a color in $C_c$. (Note that the minimum color node cut of $G$ has size smaller than $m'$, because the number of nodes in a cut of $G$ is $m'$ and some nodes have the same color. Therefore, colors of nodes incident to the other two unlabeled matchings in Fig. \ref{NPNodeCutG} cannot be in the minimum color node cut.)

  Finally, to see that the reduction can be done in polynomial time, note that $G$ has $4m'$ nodes, $2{m'}^2$ edges, and $n' + 2m'$ colors, where $n'$ and $m'$ are the number of nodes and edges in $G'$, respectively. This concludes the proof.
\end{proof}

\begin{proof}[Proof of Theorem \ref{th:npst}]
  The minimum color \st node cut problem can be reduced from the \emph{hitting set} problem. Given a universe $U$ of elements, sets $S_i$ consisting of elements in $U$ ($i = 1,2,\dots,p$), the minimum hitting set problem aims to select a minimum number of elements from $U$ such that each set $S_i$ contains at least one selected element.

  We construct a colored graph in which the minimum color $st$ node cut is identical to the minimum hitting set. Construct $p$ node-disjoint paths between an $st$ pair, each of which corresponds to a set $S_i$. If $S_i$ has $j$ elements, then its corresponding path has $j$ nodes with colors that represent the elements in $S_i$. Nodes that correspond to the same element have the same color. The reduction can clearly be done in polynomial time. A minimum color $st$ node cut contains a set of colors $C_c^{st}$ such that every path has at least one node with a color in $C_c^{st}$. This is exactly the minimum set of elements such that every set contains at least one such element.

\begin{figure}[h]
\begin{centering}
\leavevmode\includegraphics[width=0.7\linewidth]{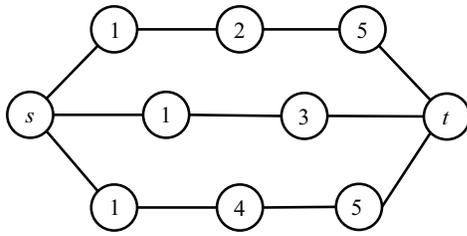}
\caption{The minimum color \st node cut is $\{1\}$. }
\label{NPSTNodeCut}
\end{centering}
\end{figure}
  We illustrate the reduction by the following example. Consider a hitting set problem where $U = \{1,2,3,4,5\}$, $S_1 = \{1, 2, 5\}, S_2 = \{1, 3\}, S_3 = \{1, 4, 5\}$. A minimum hitting set is $\{1\}$. The equivalent minimum color $st$ node cut problem is represented by Fig. \ref{NPSTNodeCut}.
\end{proof}

\bibliographystyle{IEEEtran}
\bibliography{vertexcut}

\end{document}